\documentclass[10pt]{article}

\usepackage{url}
\usepackage{amssymb}
\usepackage{amsmath}
\usepackage{algorithm}
\usepackage{algorithmic}
\newtheorem{theorem}{Theorem}

\newtheorem{definition}{Definition}

\newtheorem{claim}[theorem]{Claim}
\newcommand{\beginproof}{\medskip\noindent{\bf Proof.~}}

\newcommand{\finishproof}{\hspace{0.2ex}\rule{1ex}{1ex}}
\newenvironment{proof}{\beginproof}{\unskip\nolinebreak\finishproof\par\medskip }

\def\squareforqed{\hbox{\rlap{$\sqcap$}$\sqcup$}}
\def\qed{\ifmmode\squareforqed\else{\unskip\nobreak\hfil
\penalty50\hskip1em\null\nobreak\hfil\squareforqed
\parfillskip=0pt\finalhyphendemerits=0\endgraf}\fi}

\newcommand{\comments}[1]{}

\def\setof#1{\left\{{\let\st\colon #1 }\right\}}

  \def\calG{\mathcal{G}} 
   
 \def\calT{\mathcal{T}}

 \def\BOO={=^{\hspace{0.06cm}\epsilon}_B}

  \def\00{\mathbf{0}}
   
  \def\xx{\mathbf{x}} 
   
\def\pp{\mathbf{p}}   
 \def\xx{\mathbf{x}}  
   
 \def\ww{\mathbf{w}}  
 \def\calG{\mathcal{G}}

\begin{document}

%\mainmatter  % start of an individual contribution

\title{Ranking via Arrow-Debreu Equilibrium }

%\titlerunning{Lecture Notes in Computer Science: Authors' Instructions}

\author{Ye Du\\
\small Department of Electrical Engineering and Computer Science\\
\small University of Michigan, Ann Arbor, USA\\
duye@umich.edu}
%\thanks{Please note that the LNCS Editorial assumes that all authors have used
%the western naming convention, with given names preceding surnames. This determines
%the structure of the names in the running heads and the author index.}%
%\and Yaoyun Shi}


%\author{Ye Du}
%%
%\authorrunning{Y. Du}
%% (feature abused for this document to repeat the title also on left hand pages)
%
%% the affiliations are given next; don't give your e-mail address
%% unless you accept that it will be published
%\institute{Department of Electrical Engineering and Computer Science\\
%University of Michigan, Ann Arbor, USA\\
%duye@umich.edu}


%\toctitle{Lecture Notes in Computer Science}
%\tocauthor{Authors'Instructions}
\maketitle

\begin{abstract}
In this paper, we establish a connection between ranking theory
and general equilibrium theory. First of all, we show that the
ranking vector of PageRank or Invariant method is precisely the
equilibrium of a special Cobb-Douglas market. This gives a natural
economic interpretation for the PageRank or Invariant method.
Furthermore, we propose a new ranking method, the {\em CES
ranking}, which is {\em minimally fair}, {\em strictly monotone}
and {\em invariant to reference intensity}, but not {\em uniform}
or {\em weakly additive}. %The new CES ranking, compared with PageRank/Invariant,
%has some nonlinear property, which could be used potentially to
%find signals in a system missed those existent ranking methods.
%%\keywords{We would like to encourage you to list your keywords within
%the abstract section}
\end{abstract}

\section{Introduction}
\noindent Ranking, which aggregates the preferences of individual
agents over a set of alternatives, is not only a fundamental
problem in social choice theory but also has many applications in
real life. For instance, the well-known PageRank algorithm
\cite{bp98} is designed to rank Web pages while the Invariant
Method \cite{sv06,phv04} is proposed to evaluate the
intellectual influence of academic journals and papers. %Note that,
%mathematically, the PageRank algorithm is essentially equivalent
%to the Invariant method.

Intuitively, the PageRank and the Invariant method share a common
property in that the more ``vote" an agent gets, the higher
ranking he has. Although they work very well in practice, the
economic interpretations of their effectiveness are not obvious.
Slutzki and Volji \cite{sv06}, as well as Palacios-Huerta and
Volji \cite{phv04}, gave the first set of axioms that characterize
the Invariant method. Later, Altman and Tennenholtz \cite{at05}
gave a set of combinatorial axioms to characterize the PageRank
algorithm, while Brandt and Fischer \cite{bf07} interpreted
PageRank as a solution of a weak
tournament. %Unfortunately, those characterizations are kind of
%messy mathematically and do not provide solid economic
%explanations.

General equilibrium theory \cite{mwg} is one of the most prominent
theories in mathematical economics. It studies how a market
system, known as the ``invisible hand", makes the demands of a
market's participants equal to its supplies. Arrow and Debreu
\cite{ad54} showed that under mild conditions, a market always has
an equilibrium. The result of this research became known as the
{\em Arrow-Debreu equilibrium}.

In this paper, we will establish a connection between ranking
methods and the Arrow-Debreu equilibrium. Naturally, the
preference of one agent for another, which is usually represented
as a directed edge in a graph, can be viewed as the demand between
agents. Intuitively, the more demands a good gets, the higher
price it should have. Therefore, an equilibrium price could be a
good candidate for a ranking vector. On the other hand, the
PageRank and the Invariant method are the stationary distributions
of ergordic Markov chains. Both the existence of a PageRank or an
Invariant ranking vector and the existence of Arrow-Debreu
equilibrium can be shown via the Brouwer's fixed point theories
\cite{ad54,mwg}. We will interpret one form of a fixed point as
the other. More specifically, we will show that the ranking vector
of the PageRank or the Invariant method is indeed the equilibrium
of a special Cobb-Douglas market. To the best of our knowledge,
this is the first connection between ranking methods and the
general
equilibrium theory. %The only work, as far as we know, which has a
%similar flavor, is by Pennock and Wellman \cite{pm96}. They showed
%a mapping from a Bayesian network to a market price system with
%almost linear utility functions.
Based on our observations, we propose a new ranking method, the
CES ranking, which is minimally fair, strictly monotone, and
invariant to reference intensity, but not uniform or weakly
additive.

%{\bf Motivation:} Both the existence of the stationary
%distribution of an ergordic Markov chain and the existence of
%Arrow-Debreu equilibrium can be proved via Brower's fixed point
%theorem. We hope to encode PageRank, which is the stationary
%distribution of a special Markov chain, as the Arrow-Debreu
%equilibrium of a market. The basic idea is {\em equilibrium prices
%in a market can be used for rankings.} \cite{cmpv05}
%\\
%\\
%\noindent {\bf Main Observation:} PageRank can be interpreted as
%the equilibrium of a Cobb-Douglas market.

\section{Preliminaries}
\subsection{The Ranking Problem}
In this subsection, we will follow \cite{at05} to define {\em
ranking problems}. Let $A$ be a finite set, representing the set
of agents, and $M$ be a $|A| \times |A|$ matrix, representing the
preference relationships among the agents. A ranking problem is
represented as $\langle A,M \rangle$. %Note that, from the
%perspective of social choice theory, in our formulation of the
%ranking problem, the set of alternatives is indeed the set of
%agents.
For any $n \in \mathbb{N}$, let
$\Delta_{n}=\{(x_1,...,x_n)|\forall i, x_i \geq 0$ and $ \sum_i
x_i=1\}$.
\begin{definition} A ranking function maps a ranking problem $\langle A,M \rangle$ to a vector $\pi \in \Delta_{|A|}$.\end{definition}

\subsection{Markov chain and PageRank}
{\bf Markov chain} \cite{ks60} A discrete {\em Markov chain} is a
random process $\{X_i\}$ on a state space $S=\{s_1,...,s_n\}$ that
satisfies the {\em Markov property}:
$$P{(X_{j}|X_i,...,X_0})=P{(X_{j}|X_i})=p_{ij},$$
where $p_{ij}$ is the transition property from state $s_i$ to
$s_j$ and $\forall i, \sum_j p_{ij}=1$. Let $P$ be the transition
matrix. Correspondingly, we can define

\begin{definition} [{\bf State Transition Graph:}] Given a discrete Markov chain $\{X_i\}$, the corresponding state transition graph is $G=(V,E,W)$ where $V=S$ and
$(s_i,s_j) \in E$ iff $p_{ij}>0$, and $w_{s_i,s_j}=p_{ij}$.
\end{definition}

It is well known that if the state transition graph is strongly
connected and aperiodic, the corresponding Markov chain is called
{\em ergordic} and has a unique {\em stationary distribution}
$\pi$ such that
$$\pi=P^{T}\pi$$

\noindent {\bf PageRank} \cite{bp98} Let $G=(V, E)$ be a directed
graph with vertex set $V$ and edge set $E$. We assume that there
is no self-loop in $G$. Let $N=|V|$, and for a vertex $i\in V$,
denote by $\mathrm{out}(i)$ the out-degree of $i$. The {\em
transition matrix} of $G$ is $T=[T_{ij}]_{1\le i, j\le N}$:
\[
T_{ij}=\begin{cases}
\frac{1}{\mathrm{out}(i)} & \text {if $(i,j) \in E$}\\
0,  & \text {otherwise}\\
\end{cases}
\]
Denote by $e\in\mathbb{R}^N$ the all $1$ row vector $(1, 1,
\cdots, 1)$, and by $E\in\mathbb{R}^{N\times N}$ the all $1$
matrix. Let $\bar{T}$ be identical to $T$ except that if a row in
$P$ is all $0$, it should be replaced by $e/N$. A page without
outgoing links is called a {\em dangling} page. For some constant
$c$, $0<c<1$, the transition matrix for the PageRank Markov chain
is
$$
P=c\bar{T}+ (1-c)E/N.
$$
The PageRank $\pi$ is the stationary distribution, i.e.,$\pi
P=\pi$, of the above Markov chain $P$.

\noindent {\bf Invariant Method} \cite{sv06} In the definition of
PageRank \cite{bp98}, if the transition matrix $T$ is irreducible
(the corresponding graph is strongly connected), its unique
stationary distribution is the ranking vector. Thus, the PageRank
and the Invariant method are essentially equivalent in
mathematics.

\subsection{Arrow-Debreu equilibrium of exchange markets}
In an exchange market, there are $m$ traders and
  $n$ divisible goods. Let $\calT=\{T_1,..,T_m\}$ be the set of traders and $\calG=\{G_1,...,G_n\}$ be the set of goods. Each
  trader $i$ has an initial endowment of
  $w_{i,j} \geq 0$ of good $j$ and a utility function
  $u_{i}: \mathbb{R}_{+}^n \rightarrow \mathbb{R}$.
The individual goal of trader $T_i$ is to
  obtain a new bundle of goods that maximizes his utility.
Let
   $\mathbf{x}_{i} \in \mathbb{R}_{+}^n$ be the bundle of goods of $T_i$ after
   the exchange, where $x_{i,j}$
   is the amount of good $j$.
Naturally, the demand cannot exceed the supply:
  {$\sum_{i} x_{i,j} \leq \sum_{i} w_{i,j},$} for every good $j$.

We use $\pp\in \Delta_{n}$ to denote a price vector, where
  $p_j$ is the price of $G_j$. For any trader $T_i$, given $\pp$, we let $\xx_i^{*}(\pp)$ denote the
  bundle of goods that maximize his utility under the budget
  constraint:
\begin{equation*}
\xx_i^{*}(\pp)=\text{argmax}_{\hspace{0.08cm}\xx\in
\mathbb{R}_+^n,\
  \xx\cdot \pp\le \ww_i\cdot \pp}\  u_i(\xx).
\end{equation*}

%We use $\mathcal{X}=\{\xx_i\in \mathbb{R}_+^n:i\in [m]\}$ to
%denote an allocation of the market:
%  For each trader $T_i\in \calT$,
%  $\xx_i \in \mathbb{R}_+^n$ is the amount of goods that $T_i$ receives.
%In particular, the amount of $G_j$ that $T_i$ receives in
%$\mathcal{X}$ is
%  $x_{i,j}$.

\begin{definition} [{\bf Arrow-Debreu equilibrium}]
A market equilibrium is a price vector $\pp\in \Delta_{n}$ such
that the market clears:
\begin{quote}
For every good $G_j\in \calG$, $\sum_{i\in [m]} x_{i,j}^{*}(\pp)
 \le \sum_{i\in [m]} w_{i,j};$ If $p_j>0$, then $\sum_{i\in [m]}
x_{i,j}^{*}(\pp)  = \sum_{i\in [m]} w_{i,j}.$
\end{quote}
\end{definition}

%\begin{definition}[Economy Graph]
%Given an exchange market, we define a directed graph
%  $G=(\calT,E)$ as follows.
%The vertex set of $G$ is exactly $\calT$, the set of traders in
%the market. For every
%  two traders $T_i\ne T_j\in \calT$, we have an edge from $T_i$ to
%  $T_j$ if $T_j$ possesses a good for which $T_i$ has demand. \footnote{The direction of the edge in the economy graph here is
% different from the definitions of economy graphs in \cite{cmpv05,m97}.}
%  $G$ is
%called the \emph{economy graph} of the market. We say the market
%is \emph{strongly connected} if $G$ is \emph{strongly connected}.
%\end{definition}

\subsection{CES Utility Functions}

{\bf CES utility functions:} \cite{cmpv05,mwg} The CES (Constant
Elasticity of Substitution) function over a bundle of goods
$(x_{i1},...,x_{in})$ is the family of utility functions
$u_i(x_{i1},...,x_{in})$
$=(\sum_{j=1}^{n}\alpha_{ij}x_{ij}^{\rho_i})^{1/\rho_i}$, where
$-\infty < \rho_i <1$, $\rho_i \neq 0$ and $\alpha_{ij} \geq 0$.
The parameter $1/(1-\rho_i)$ is called the {\em elasticity of
substitution}. The CES utility function has a very nice property:
its demand functions have explicit forms. That is, given a
strictly positive price vector $\pi \in \mathbb{R}_{++}^{n}$, the
demand $x_{ij}$ is
$$x_{ij}=\frac{\alpha_{ij}^{1/(1-\rho_i)}}{\pi_j^{1/(1-\rho_i)}}\times \frac{\sum_{k}\pi_kw_{ik}}{\sum_{k}\alpha_{ik}^{1/(1-\rho_i)}\pi_k^{-\rho_i/(1-\rho_i)}}$$
There are three important utility functions within the CES
category.
\begin{enumerate} \item $\rho_i \rightarrow 1$ corresponds to linear utility functions, where $u_i(x_{i1},...,x_{in})=\sum_{j}\alpha_{ij}x_{ij}$. In this case, the
set of goods that the agent wants are perfect substitutes for each
other. \item $\rho_i \rightarrow -\infty$ corresponds to Leontief
utility functions. The Leontief utility function, in general, has
the form of $u_i(x_{i1},...,x_{in})=\min\limits_{j:
\beta_{ij}>0}\frac{x_{ij}}{\beta_{ij}}$, where $\beta_{ij} \geq
0$. In this case, the set of goods that the agent wants are
perfect complement of each other. \item $\rho_i \rightarrow 0$
corresponds to Cobb-Douglas utility functions. The Cobb-Douglas
utility function, in general, has the form of
$u_i(x_{i1},...,x_{in})=\prod_{j}x_{ij}^{\beta_{ij}}$, where
$\beta_{ij} \geq 0$. This demand function is a perfect balance of
substitution and complementarity effects \cite{cmpv05}.
\end{enumerate}

\section{PageRank/Invariant Method V.S. a Cobb-Douglas Market}
In this section, we will establish a connection between
PageRank/Invariant Method and Arrow-Debreu equilibrium. We do this
by showing a more general theorem about Markov chains.
\begin{theorem} \label{thm:main} Given an ergordic Markov chain, there is a
mapping from the Markov chain to a Cobb-Douglas market, such that
the stationary distribution of the Markov chain is precisely the
Arrow-Debreu equilibrium of the Cobb-Douglas market.\end{theorem}
\begin{proof} The general idea of the proof is to reduce the state transition graph of an ergordic Markov chain to the
economy graph of a Cobb-Douglas market. Given the state transition
graph $G=(V,E,W)$ and the transition matrix $P$, we can construct
a Cobb-Douglas economy graph as follows: for $i \in [1..n]$, there
is a trader $T_i$ corresponding to each state $s_i$, and there is
a directed link from $T_i$ to $T_j$ iff $(s_i,s_j) \in E$; for
trader $T_i$, let $N(T_i)$ be the set of outgoing neighbors of
$T_i$. The utility function of $T_i$ is $u_i(x_i)=\prod_{j \in
N(T_i)}x_{ij}^{p_{ij}}$, where $p_{ij}$ is the transition
probability. Initially $T_i$ has one unit of the good $G_i$ but no
other goods. We call such a Cobb-Douglas economy $M$. We claim
that the market equilibrium of $M$ is also the stationary
distribution of $P$.

First of all, since $G$ is strongly connected and the Cobb-Douglas
utility function belongs to the CES utility function class,
according to Theorem 1 of Codenotti {\em et al.} \cite{cmpv05} $M$
has a strictly positive equilibrium. By the demand function of CES
utility function, when $\rho_i \rightarrow 0$,
$$x_{ij}=\frac{p_{ij}\pi_i}{\pi_j}.$$
By the definition of Arrow-Debreu equilibrium, for every good with
strictly positive price, its demand must be equal to its supply.
Thus,
$$\sum_{i}x_{ij}=\sum_{i}\frac{p_{ij}\pi_i}{\pi_j}=1.$$
Equivalently,
$$\sum_{i}p_{ij}\pi_i=\pi_j.$$
Thus, $\pi$ is the stationary distribution of the Markov chain
$P$.
\end{proof}
Actually, by the above reduction, we also implicitly show that $M$
has a unique equilibrium. Most importantly, since the PageRank or
Invariant method is a special Markov chain, the ranking vector of
the PageRank or the Invariant method can also be interpreted as
the
equilibrium of a Cobb-Douglas market.\\

\noindent {\bf Remarks:} Eaves \cite{eaves85} showed that the
computation of an equilibrium for a Cobb-Douglas market can be
reduced to solving a linear equation system. Although it was not
explicitly claimed in \cite{eaves85}, Eaves's result implies that
an equilibrium of a Cobb-Douglas market actually corresponds to a
principle eigenvector of a stochastic matrix. In Theorem
\ref{thm:main}, we show the reverse direction of Eaves' reduction.
That is a principle eigenvector of a stochastic matrix corresponds
to an equilibrium of a special Cobb-Douglas market. Thus, our
result, which is essentially different from Eaves' in terms of
motivations, is a complement of the result in \cite{eaves85}.\\

\noindent {\bf Economic Interpretations:} It is believed that the
validity of PageRank comes from the fact that the Markov chain is
a good model for the Web surfing behavior of Web users. In web
graph, a link from page $p$ to page $q$ means that a Web user at
page $p$ may find the content of page $q$ is useful. Thus, a link
in web graph means a vote or reference. Intuitively, the more
votes a page gets, the more important it is. Indeed, for a Web
user, his goal is to maximize his information needs by following
outgoing links of a page to visit other pages. Thus, in our
Cobb-Douglas economy graph, each Web page is corresponding to an
agent, the content of the page is corresponding to the good the
agent initially owns, and a link from $p$ to $q$ means that the
agent on page $p$ has a demand for the content of page $q$.
Intuitively, the more ``demand" a page gets, the more important it
is. %Actually the demand could be interpreted in Web graph in two
%ways: on one hand, the content of page $p$ may be similar to the
%content of page $q$, thus from the perspective of Web users, they
%are substitutes of each other; on the other hand, the content of
%page $p$ could be a complement of another page $q$.
% By changing the utility
%functions of each traders from the spectrum of CES utility
%functions, we may introduce a new series of ranking methods.

Theorem \ref{thm:main} provides a new perspective to view
PageRank. That is the substitution and complementarity effects of
outgoing links. For instance, suppose we have a directory page of
a university, which has outgoing links pointing to the home pages
of each of the unversity's the departments. If a Web user clicks
one of the outgoing links, it is unlikely that he will click any
other. Thus, for this page, its outgoing links are more likely to
be substitutes for each other than complements. On the other hand,
suppose we have a news page, which has outgoing links pointing to
related news pages. A Web user who clicks one of the outgoing
links is likely to click another one. Thus, for this page, its
outgoing links are more likely to be complements for each other
than substitutes. By Theorem \ref{thm:main} and the properties of
the Cobb-Douglas utility function, the set of pages that a Web
page points to is a mix of the substitution and complementarity
effects with elasticity of substitution $1$ in PageRank.

\section{Ranking via Arrow-Debreu Equilibrium}
The Cobb-Douglas utility function corresponds to the CES utility
function with $\rho \rightarrow 0$. Thus, by choosing CES utility
functions with different elasticities, we can naturally extend the
idea of PageRank to a new spectrum of ranking algorithms. We
propose the ranking method, which is called {\em CES ranking},
below.
\begin{algorithm}
%\begin{algorithmic}
\caption{CES ranking} 1. Given the agents set $A$, choose a CES
utility function $u_i$ for each agent $i \in A$ and set the
initial endowment $w_i$ of it as $\forall j \neq i, w_{ij}=0$ but
$w_{ii}=1$. The new ranking problem is $\langle A,
\{\alpha_{ij}\}_{i \in A, j \in A}, \{\rho_i\}_{i \in A}, \{w_i|i \in A \} \rangle$. W.L.O.G., for any $i$, let $\sum_j \alpha_{ij}=1$. \\
2. %Construct the economy graph $G$ and make it strongly connected
%by modifying the utility functions of agents as PageRank.
If for agent $i$, $\alpha_{ij}=0$ for all $j$, set
$\alpha_{ij}=1/|A|$ for each $j$. Hence, $u_i=(\sum_{j=1}^{n}(1/|A|)x_{ij}^{\rho_i})^{1/\rho_i}$.\\
3. For each agent $i$, update $\alpha_{ij}$ to be
$\alpha_{ij}*\beta+(1/|A|)*(1-\beta)$ for every $j$, where
$\beta=0.85$. Correspondingly, the updated utility function is $u_i=(\sum_{j=1}^{n}(\alpha_{ij}*\beta+(1/|A|)*(1-\beta))x_{ij}^{\rho_i})^{1/\rho_i}$\\
4. Construct the economy graph $G$ with respect to the CES economy defined above. It is easy to see that $G$ is strongly connected.\\
5. Compute an equilibrium of $G$ and use it as the ranking vector.
%\end{algorithmic}
\end{algorithm}

Note that the new CES ranking problem $\langle A,
\{\alpha_{ij}\}_{i \in A, j \in A}, \{\rho_i\}_{i \in A}, \{w_i|i
\in A \} \rangle$ is a generalization of the ranking problem
$\langle A,M \rangle$ defined in the Preliminaries. Now we discuss
the existence, uniqueness, and efficiency of the CES ranking, as
well as some other properties related to ranking.

\noindent {\bf Existence of a ranking vector:} By Theorem 1 in
\cite{cmpv05}, as long as the economy graph $G$ is strongly
connected, there is always a strictly positive equilibrium. By the
definition of the CES ranking, it is obvious that the economy
graph of it is strongly connected. Thus, a ranking vector always
exists.

\noindent {\bf Uniqueness:} According to \cite{cmpv05}, for CES
utility functions with $-1 \leq \rho <1$, the set of equilibria is
convex. We further show that:
\begin{claim} The CES ranking has a unique ranking vector if $\forall i, \rho_i \geq 0$. \end{claim}
In order to prove this claim, we first introduce the definition of
gross substitute (GS).
\begin{definition} \cite{mwg} For any $j$, let $z_j=\sum_i x_{ij}-\sum_{i}w_{ij}$ be the excess demand function for $G_j$. The function $z(.)$ has the
gross substitute property if whenever $\pi^{'}$ and $\pi$ are such
that, for some $l$, $\pi^{'}_{l} > \pi_{l}$ and $\pi^{'}_{j} =
\pi_{j}$ for $j \neq l$, we have $z_{j}(\pi^{'})
> z_{j}(\pi)$ for $j \neq l$.\end{definition}
If in the above definition the inequalities are weak, the property
is referred to as weak gross substitute (WGS). It is well known
that:
\begin{theorem} \cite{mwg} If the aggregate excess demand function of an exchange economy has the GS property, the economy has at most one equilibrium.\end{theorem}
Now we can prove the above claim.\\
\begin{proof}
\footnote{It is well known \cite{cmpv05} that the CES utility
functions satisfy WGS when $\rho \geq 0$. However, WGS does not
imply the uniqueness of equilibrium.} In the CES ranking, by the
Theorem 1 and Lemma 1 in \cite{cmpv05}, as long as the economy
graph $G$ is strongly connected, an equilibrium exists and every
equilibrium is strictly positive. Suppose we have two equilibria
$\pi^{'}$ and $\pi$ such that for $l$, $\pi^{'}_{l} > \pi_{l}$ and
$\pi^{'}_{h} = \pi_{h}$ for $h \neq l$. Note that $\forall i, j,
\alpha_{ij}>0$. Thus, for any $j \neq l$,
\begin{eqnarray*} \sum_i x_{ij}&=&\sum_i\frac{\alpha_{ij}^{1/(1-\rho_i)}}{\pi_j^{1/(1-\rho_i)}}\times \frac{\pi_i}{\sum_{k}\alpha_{ik}^{1/(1-\rho_i)}\pi_k^{-\rho_i/(1-\rho_i)}}\\
&=& \sum_i \alpha_{ij}^{1/(1-\rho_i)}\times
\frac{\pi_i/\pi_j}{\sum_{k}\alpha_{ik}^{1/(1-\rho_i)}(\pi_k/\pi_j)^{-\rho_i/(1-\rho_i)}}\\
&<&\frac{\alpha_{lj}^{1/(1-\rho_l)}\times
\pi_l^{'}/\pi_j^{'}}{\sum_{k}\alpha_{lk}^{1/(1-\rho_l)}(\pi_k^{'}/\pi_j^{'})^{-\rho_l/(1-\rho_l)}}+\sum_{i
\neq l}
\frac{\alpha_{ij}^{1/(1-\rho_i)}\times \pi_i^{'}/\pi_j^{'}}{\sum_{k}\alpha_{ik}^{1/(1-\rho_i)}(\pi_k^{'}/\pi_j^{'})^{-\rho_i/(1-\rho_i)}}\\
&=& \sum_i x_{ij}^{'}\end{eqnarray*} Thus, in the CES ranking, the
excess demand function has the GS property. Therefore, its
equilibrium and the ranking vector are unique.
\end{proof}

However, when $\rho<-1$, there may be multiple equilibria sets.
Actually, for ranking problems, we do not have to insist on the
uniqueness of ranking vectors. That is because in most cases,
there are different ranking criteria for one ranking problem. It
is not surprising that different criteria induce different
rankings.

\noindent {\bf Efficiency of Computation:} When $-1 \leq \rho<1$,
an equilibrium of a CES market can be computed via convex
programming \cite{cmpv05}. Thus, it is in polynomial time.
However, for some special utility functions (such as the Leontief
utilty function), it is PPAD-hard \cite{csvy08} to compute an
equilibrium of it. However, for ranking problems that have
relatively small sizes or do not require the real time computation
of ranking vectors, the efficiency may not be a serious concern.

In the next section, we study five natural properties, which are
satisfied by the PageRank and the Invariant method, with respect
to the CES ranking. First, we extend the concept of minimally fair
\cite{at07} to the CES ranking.
\begin{definition} A ranking system is minimally fair if when for any $i,j$, $\alpha_{ij}=0$, the
ranking score of agent $i$ equals that of $j$ for agent $i,j \in
A$.
\end{definition}
\begin{claim} The CES ranking is minimally fair. \end{claim}
\begin{proof} Since initially $\alpha_{ij}=0$ for any $i,j$, in order to make the economy graph strongly connected,
we set the utility function for each agent as
$u_i=((\sum_j\frac{1}{n}x_{ij}^{\rho})^{1/\rho}$. With this setup,
by the market clearing condition, we get
$$ \sum_i\frac{(1/n)^{1/(1-\rho)}}{\pi_j^{1/(1-\rho)}}\times \frac{\pi_i}{\sum_{k}(1/n)^{1/(1-\rho)}\pi_k^{-\rho/(1-\rho)}}=1 $$
Thus, $\forall j$,
$$\pi_j=(1/(\sum_k\pi_k^{-\rho/(1-\rho)}))^{1-\rho}.$$
Note that the right-hand size of the above equation is independent
of $j$. Thus, $\pi=(1/n,...,1/n)$ is the only equilibrium for the
market. Therefore, the CES ranking is minimally fair.\end{proof}

Next, we extend the strictly monotone definition in \cite{at07} to
the CES ranking.
\begin{definition} A ranking system is {strictly monotone} iff for any two agents $i$ and $j$,
if for any other agent $p$, $\alpha_{pi} \leq \alpha_{pj}$ and
there exists $h$ such that $\alpha_{hi} < \alpha_{hj}$, the
ranking score of agent $i$ is strictly less than that of $j$.
\end{definition}

\begin{claim} The CES ranking is strictly monotone when the utility functions of all the agents have the same elasticity of substitution. \end{claim}
\begin{proof} By the market clearing condition, for any two agents
$i$ and $j$,
\begin{eqnarray*}
\pi_i^{1/(1-\rho)}&=&\sum_{p}\frac{\alpha_{pi}^{1/(1-\rho)} \times
\pi_p}{\sum_{k}\alpha_{pk}^{1/(1-\rho)}\pi_k^{-\rho/(1-\rho)}}\\
&<&\sum_{p}\frac{\alpha_{pj}^{1/(1-\rho)} \times
\pi_p}{\sum_{k}\alpha_{pk}^{1/(1-\rho)}\pi_k^{-\rho/(1-\rho)}}\\
&=&\pi_j^{1/(1-\rho)}
\end{eqnarray*}
Thus, $\pi_i<\pi_j$.
\end{proof}
Actually, the property of strictly monotone corresponds to the
intuition that the more demands a good gets, the higher price it
is.

Slutzki and Volij showed \cite{sv06} that,

\begin{theorem}  \cite{sv06} If a ranking system satisfies {\em uniform}, {\em weakly
additive} and {\em invariant to reference intensity}, the ranking
system must be the Invariant method. \end{theorem}

In the next section, we will study the relationship between the
CES
ranking and the three properties above. %First of all, we give the
%definition of those properties.

\begin{definition} \cite{sv06} A ranking problem is regular if $\forall i,j, \sum_k \alpha_{ik}=\sum_k
\alpha_{jk}$ while $\forall i,j, \sum_k \alpha_{ki}=\sum_k
\alpha_{kj}$. A ranking function is uniform if for every regular
ranking problem, the ranking score of each agent is $\frac{1}{N}$
where $N$ is the number of agents.
\end{definition}
\begin{claim} The CES ranking is not uniform. \end{claim}
\begin{proof}
Suppose a system has three agents while the parameters of the
agents are $\rho_1=\rho_2=\rho_3=\frac{1}{2}$ and
$\alpha_{11}=\frac{1}{3}$, $\alpha_{12}=\frac{1}{3}$,
$\alpha_{13}=\frac{1}{3}$, $\alpha_{21}=\frac{5}{12}$,
$\alpha_{22}=\frac{1}{6}$, $\alpha_{23}=\frac{5}{12}$,
$\alpha_{31}=\frac{1}{4}$, $\alpha_{32}=\frac{1}{2}$,
$\alpha_{33}=\frac{1}{4}$. By the market clearing condition,
\begin{eqnarray*}
\pi_1^{2}&=&\frac{\frac{1}{3}^{2} \times
\pi_1}{\frac{1}{3}^{2}\pi_1^{-1}+\frac{1}{3}^{2}\pi_2^{-1}+\frac{1}{3}^{2}\pi_3^{-1}}+\frac{\frac{5}{12}^{2}
\times
\pi_2}{\frac{5}{12}^{2}\pi_1^{-1}+\frac{1}{6}^{2}\pi_2^{-1}+\frac{5}{12}^{2}\pi_3^{-1}}+\frac{\frac{1}{4}^{2}
\times
\pi_3}{\frac{1}{4}^{2}\pi_1^{-1}+\frac{1}{2}^{2}\pi_2^{-1}+\frac{1}{4}^{2}\pi_3^{-1}}
\end{eqnarray*}
It is easy to check that
$\pi=(\frac{1}{3},\frac{1}{3},\frac{1}{3})$ cannot satisfy the
above equation. Thus, the CES ranking is not uniform.
\end{proof}

Actually, the uniform property requires that the ranking score of
an agent is linearly proportional to the number of ``votes" it
gets. This assumption may not be reasonable universally. For
instance, in the citation analysis, suppose both paper A and B
have m citations. In an extreme case, the citations of paper A may
only come from one research group while the citations of B come
from different research groups. Intuitively, paper B should be
more important than paper A. However, any ranking algorithm with
the uniform property cannot distinguish those two cases. The CES
ranking, as a nonlinear ranking method, may have the potential to
find out new signals that were missed by the uniform ranking
methods.

The {\em weakly additive} \cite{sv06} property says that for a
regular ranking problem, the ranking score is still linearly
proportional to the ``votes" after a symmetric perturbation. Since
the CES ranking is not uniform, it cannot satisfy the weakly
additive property, either. %Again, the weakly additive property is
%a stronger form of linearity requirement for a ranking system
%while the CES ranking is nonlinear.

\begin{definition} \cite{sv06} A ranking system is invariant to reference intensity if for any agent $i$,
when we multiply $\alpha_{ij}$ by a positive constant $\lambda$
for every $j$, it cannot change the ranking score of any agent.
\end{definition}
\begin{claim} The CES ranking is {\em invariant to reference intensity}.\end{claim}
\begin{proof} Note that
\begin{eqnarray*}
x_{ij}&=&\frac{\alpha_{ij}^{1/(1-\rho_i)}}{\pi_j^{1/(1-\rho_i)}}\times
\frac{\sum_{k}\pi_kw_{ik}}{\sum_{k}\alpha_{ik}^{1/(1-\rho_i)}\pi_k^{-\rho_i/(1-\rho_i)}}\\
&=&\frac{(\lambda\alpha_{ij})^{1/(1-\rho_i)}}{\pi_j^{1/(1-\rho_i)}}\times
\frac{\sum_{k}\pi_kw_{ik}}{\sum_{k}(\lambda\alpha_{ik})^{1/(1-\rho_i)}\pi_k^{-\rho_i/(1-\rho_i)}}\\
&=&\frac{\alpha_{ij}^{1/(1-\rho_i)}}{\pi_j^{1/(1-\rho_i)}}\times
\frac{\sum_{k}\pi_kw_{ik}}{\sum_{k}\alpha_{ik}^{1/(1-\rho_i)}\pi_k^{-\rho_i/(1-\rho_i)}}
\end{eqnarray*}
Thus, multiplying $\alpha_{ij}$ by a positive constant $\lambda$
for every $j$ cannot change the demand function. Therefore, the
set of equilibria remains the same.\end{proof}

When we summarize the above claims together, we get
\begin{theorem} \label{thm:property} The CES ranking is {\em minimally fair}, {\em strictly monotone} and {\em invariant to reference intensity}, but not {\em uniform} or {\em weakly additive}. \end{theorem}

\section{Conclusion and Future Works}
In this paper, we have established a connection between the
ranking theory and the general equilibrium theory. First, we
showed that the ranking vector of PageRank or Invariant method is
actually the equilibrium of a special Cobb-Douglas market. This
gives a natural economic interpretation for the PageRank and
Invariant method. Furthermore, we propose a new ranking method,
the CES ranking, which is minimally fair, strictly monotone, and
invariant to reference intensity, but not uniform or weakly
additive. The new CES ranking, compared to PageRank and the
Invariant method, is nonlinear, and could be potentially used to
find signals in a system missed by those existing ranking methods.

With the observations in this paper, we have a complete picture of
the encoding power of the three limiting cases of CES utility
functions. Pennock and Wellman \cite{pm96} showed that economies
with almost the linear utility functions can encode Bayesian
networks. Codenotti {\em et al.} \cite{csvy08} proved that
economies with the Leontief utility functions can encode bimatrix
games. Now we demonstrate that economies with the Cobb-Douglas
utility functions can encode Markov chains.

We believe that this paper points to a few promising directions
that are worth further exploration.
\begin{itemize} \item Explore more properties that the CES ranking satisfies and make justifications for the properties it does not satisfy.
\item For various applications, what is the ``right" utility
function for each agent? We may go beyond the CES utility
functions and explore other ones, such as WGS utility functions
\cite{cmpv05}. \item Design efficient algorithms to compute
ranking vectors. \item Further investigate the uniqueness of
ranking vectors. If there are multiple equilibria points, do they
induce the same ranking? If not, interpret their different
economic meanings in the context of ranking. \item Last but not
least, design an effective evaluation system for ranking methods
and find an application where the CES ranking can outperform
existing ranking methods.
\end{itemize}

\section*{Acknowledgements}
I would like to thank Xiaotie Deng, Romesh Saigal, Yaoyun Shi and
Michael Wellman to read the preliminary version of this paper and
provide valuable suggestions.

%\bibliographystyle{plain}
%\bibliography{ranking}

\end{document}